\documentclass{article}
\usepackage[utf8]{inputenc}
\usepackage[english]{babel}
\usepackage{calrsfs}
\usepackage{amsthm, amsmath, amssymb}
\usepackage{authblk}
\usepackage[symbol]{footmisc}
\usepackage{hyperref}
\usepackage{xcolor}

\usepackage{amsmath}
\newtheorem{prop}{Proposition}
\newtheorem{axiom}{Axiom}
\newtheorem{remark}{Remark}

\newtheorem{example}{Example}[section]
\usepackage{mathtools}          
\DeclarePairedDelimiter\Floor\lfloor\rfloor
\DeclarePairedDelimiter\Ceil\lceil\rceil
\usepackage{graphicx}
\usepackage{tikz}

\hypersetup{
    colorlinks=false,
    linkcolor=black,
    filecolor=black,      
    urlcolor=black,
}
\theoremstyle{definition}
\newtheorem{definition}{Definition}[section]

\title{Proportionality in Committee Selection with Negative Feelings}
\author[1]{Nimrod Talmon\thanks{talmonn@bgu.ac.il}}
\author[2]{Rutvik Page\thanks{rutvikpage1999@gmail.com}}
\affil[1]{Ben-Gurion University, Israel}
\affil[2]{Indian Institute of Information Technology, Nagpur, India}

\date{}

\begin{document}

\maketitle

\begin{abstract}
We study a class of elections in which the input format is trichotomous and allows voters to elicit their negative feelings explicitly. In particular, we study multiwinner elections with a special proclivity to elect proportionally representative committees. That is, we design various axioms to deal with negative feelings and suggest some structures to these preferences that allow better preference aggregation rules. We propose two different classes of axioms designed to aggregate trichotomous preferences more efficiently. We propose trichotomous versions of some well known multiwinner voting rules and report their satisfiability of our axioms. Hence, with reports of our simulations as evidence, we build upon the social optimality of our proportionality based axioms to evaluate the quality of voting rules for electing a proportionally representative committee with trichotomous ballots as inputs. 
\end{abstract}

\section{Introduction}

Aggregation of possibly conflicting preferences to take socially optimal decisions is a central problem in the field of Computational Social Choice and has been actively pursued by researchers in the AI community \cite{Conitzer}. One of the important scenarios in this field is to elect a committee from a given set of alternatives. In multiwinner elections, the challenge is to select a size-$k$ committee from a set of candidates, given a voter preference profile. Generally, the preference elicitation methods available to voters are either ranked ballots or approval ballots. In ranked ballots, the voters put forth complete and strict rankings over the candidates in a descending order where the most preferred candidate is ranked the highest and the least preferred candidate is ranked the lowest. Another popular way of preference elicitation is approval ballots, in which each voter only puts forth a subset of candidates that she approves. We study a more general setting of preference elicitation in which  each voter divides the set of candidates into three mutually non intersecting subsets. These three subsets contain the candidates about whom the voters are in approval of, are indifferent about or are in disapproval of respectively. 

The design of multiwinner voting rules is a challenging task since its applications range from excellence based rules
\cite{Nimrod1,barbera2008choose,propMultiwinner,articleProp}  through selecting a diverse set of candidates \cite{propMultiwinner,articleProp,SkowronLang} to proportional representation \cite{LB15,Lu10budgetedsocial,sfelkind,JR,PJR,CCAV}.
Recently, the design of efficient multiwinner voting rules and defining their natural social choice properties \cite{PJR,betzler2013computation,Aziz2014ComputationalAO} has received a considerable attention from the Artificial Intelligence community. As a consequence, a rich variety of multiwinner voting rules and social choice axioms are emerging. 

In various preference aggregation scenarios some voters might wish not only to 
describe their positive feelings towards certain alternatives, but 
also to describe their negative feelings with the same rigour while preserving the right to remain neutral about some \cite{Alcantud,Hillinger}. One prominent example is the veto rule in single winner elections, where each voter specifies one 
candidate which she dislikes. Another extreme example is approval voting, 
in which each voter provides an \emph{approval set}, which is simply a 
subset of candidates from the available candidates; here, it is not 
clear whether the voter simply does not care about the candidates she 
did not include in her approval set or perhaps she has, say, strong 
negative feelings towards them.

For instance, inclusion of certain vegetables in grocery 
shopping list does not necessarily mean that those items that are not 
on the list are despised by the shopper, neither does it necessarily 
mean that the shopper is indifferent about them. \emph{Au contraire},
while deciding to go to the movies with friends, its not necessary that
everyone has a binary preference over every movie being 
considered, i.e., there might be some friends in the group who are 
indifferent about watching some movie they neither like or dislike or 
there may be others who particularly dislike a movie hence emphasizing 
the fact that there is a clear distinction between indifference and disliking, which is inadequately captured by approval voting.
Thus, a method of preference elicitation that widens the scope of 
expressibility of the voter requirements is required in order to 
capture social choice optimally. 

Political polarization is a common phenomenon in the current era as a diverse mixture of opinions float via the internet to the voters. Approval ballots (also referred to as dichotomous preferences) provide a way of either approving or disapproving a particular candidate hence leaving the populace to a more vulnerable position with regards to political polarization. Trichotomous ballots provide a way out of this by allowing voters to show their apathy towards some candidates and hence to some extent reduce the magnitude of political polarization.

\subsection{Our Contributions}
In this paper, we initiate a principled study for allowing voters to directly express negative feelings towards certain available options. In the standard model of social choice there are various input formats, where the most prominent ones are perhaps  Approval ballots,  Ordinal ballots and Cumulative ballots (Plurality voting is very popular, however we view it as a kind of approval voting).
We propose two broad classes of axioms for evaluating committees selected by elections having input as trichotomous preferences basically differentiated by the liberality of axioms to provide representation to certain voters. Therefore, we propose new axioms for formalizing the view of proportional representation under trichotomous preferences and explore the nature of different voting rules in trichotomous domains. We propose some voting rules tailored for trichotomous inputs and show by simulations the extent to which our rules as well as known polynomial time multiwinner voting rules satisfy our axioms and also formalize a preference domain restriction in which some of these axioms seem to be a natural fit (Definition 3.1).

\subsection{Related Work}

The works of Condorcet (1793) were the first ones to propose a voting rule in which the voters were required to partition the candidate set into three groups according to their preferences. Brams \cite{brams} proved that elections allowing only approval votes are equivalent to those allowing only disapproval votes and that disapproval votes are redundant when there is an absence of lower bound on the number of approvals a candidate should muster to be declared the winner. Yilmaz \cite{yilmaz} presents a normative study of trichotomous voting as a superior alternative to approval voting in a way that trichotomous ballots represent voter feelings more accurately. Falsenthal et al. \cite{Falsenthal}  initiate the study of single winner election under trichotomous preferences and present a contrast between approval voting and voting with trichotomous ballots and conclude that the latter leads to voters being more decisive. Hillinger et al.\cite{Hillinger}, Alcantud et al \cite{Alcantud}. and Smaoui et al \cite{smaoui} and Lapresta et al. \cite{bestworst} propose and study utilitarian scoring rules and their axiomatic properties which in a way allow voters to attribute ranks to candidates over a range (voters assign scores to candidates from a given range) typically selecting that candidate as the winner which has the highest positive difference between the number of voters approving and disapproving the candidate. 
Baumeister et al. \cite{baum} and Zhou et al. \cite{Zhou} study the utilitarian and egalitarian variants of voting rules for committee elections with voter dissatisfaction with trichotomous preferences. While the former study concentrates on using a distance based approach like Kemeny Distance, the latter generalizes popular voting rules like Chamberlin-Courant Rule, Proportional Approval Voting and Satisfaction Approval Voting to trichotomous domains and find the parameterized complexity of winner determination under these rules. 
Ouafdi et al. \cite{Ouafdi} draw out a comparison between evaluative voting (trichotomous preferences with scores being assigned from the set \{2, 1, 0\} to each candidate in contrast to the earlier stated \{1, 0, -1\}) and popular voting methods like Borda Rule, Plurality and the Approval Rule and investigate the proclivity of evaluative voting to elect Condorcet committees. Aziz and Lee \cite{aziz2020expanding}  generalize Proportional Solid Coalitions defined for strict preferences to weak preferences and show that Proportional Justified Representation \cite{PJR} is also a specialization of the same.

\section{Preliminaries}

We first provide preliminaries regarding approval ballots and our modeling of negative feelings via a generalization to trichotomous preference and then provide preliminaries regarding proportionality axioms for approval ballots that, later, we adapt to our setting.

\subsection{Dichotomous and Trichotomous Preferences}

Given a set of alternatives of size $m$, $A = \{a_1, \dots, a_m\}$, in approval voting each voter $i$ in voter set $V$ of size $n$ specifies an \emph{approval set} $A_i \subseteq A$ and the goal is to select a committee $W$ of size exactly $k$. Usually, the alternatives in $A_i$ are understood as the alternatives ``approved'' by $i$.
Define a dichotomous voter profile as a vector $A_{Dic} = (A_1, \dots A_n)$ such that $A_i \subseteq A$ $\forall i \in V$.
Informally, the alternatives in $A \setminus A_v$ can be understood either as (1) alternatives for which $v$ does not have any feelings about, or as (2) alternatives for which $v$ has negative feelings about (contrasted with the alternative in $A_v$, for which it can be understood that $v$ has positive feelings about).
In a way, this ambiguity is precisely the problem, as there is no way for the aggregation mechanism to figure out which of the two cases it is for each voter.  

A natural remedy might be to let each voter specify not only an approval set, but also a disapproval set; that is, let each voter $i$ specify $A^+_i$ as well as $A^-_i$, such that $A_i^+, A_i^- \subseteq A$ and $A_i^+ \cap A_i^- = \emptyset$.
Then, for each alternative, a voter $i$ would place the alternatives for which she has positive feelings about in $A^+_i$ and place the alternatives for which she has negative feelings about in $A^-_i$; the remaining alternatives,i.e., those in $A^0_i := A \setminus (A^+_i \cup A^-_i)$ are those for which $i$ has not feelings at all. We therefore define a trichotomous preference profile as a vector $A_{tri} = ((A_1^+, A_1^-), \dots (A_n^+, A_n^-))$.

\subsection{Proportionality for Approval Ballots}

We recall two known axioms of proportional representation in approval voting for selecting a committee of size $k$.
\begin{definition}{JR \cite{JR}}
A committee $W$ is said to satisfy \emph{Justified Representation} if there does not exist group of voters $V' \subseteq V$ such that $|V'| \geq \frac{n}{k}$ and 
$(\cup_{i \in V'}A_i) \cap W = \emptyset$ 
\end{definition}
\begin{definition}{PJR \cite{PJR}}
A committee $W$ is said to satisfy \emph{Proportional Justified Representation} if there does not exist a group of voter $V' \subseteq V$ with size $|V'| \geq l\frac{n}{k}$ for $l \in \{1, 2, \dots k\}$ such that $|\cap_{i \in V'}A_i| \geq l$ but 
$|(\cup_{i \in V'}A_i) \cap W| < l$.
\end{definition}
The idea behind \emph{JR} and \emph{PJR} is that if voters in large enough groups are inclined to have similar choices, then at least some voters of the group should get some representation in the committee $W$. 
In this paper, we generalize dichotomous preferences to trichotomous preferences which allows an extension in expressiveness to the voters by allowing them not only to elicit approval or disapproval but also indifference over the candidates in the candidate set. We represent the position of a candidate $c$ for a voter $i$ such that $pos_c(i) \in \{1, 0, -1\}$ if the candidate lies in the approval, indifference and disapproval set respectively of voter $i$. We define the positional score of a candidate $c$ as the sum of positions of the candidate in the preference ballots of all voters in the electorate i.e. $\Sigma_{i \in V}pos_c(i)$. Moreover, we use $[k]$ as an abbreviation for the set $\{1, 2, \dots k\}$.  
\begin{remark}\label{remark:remark1}
A dichotomous approval voter profile \textbf{$A_{dic}$} = $(A_1, A_2, \dots A_n)$ is essentially a trichotomous voter profile \textbf{$A_{tri}$} = ~$((A_1^+, A_1^-), \dots (A_n^+, A_n^-))$ with $A_i^+ \cup A_i^- = C \text{ } \forall i \in V$, or $A_i^0 = \emptyset$ $\forall i \in V$ which means that a trichotomous profile with every voter casting an empty indifference ballot is essentially a dichotomous profile. 
\end{remark}

\section{Axioms for Trichotomous Preferences}

The basic unit of voters whom a proportionally representative committee might proffer representation to would be a group of voters who have the proclivity to have similarly aligned preferences. Ideally, this calls for the pursuit of providing representation to every large enough and seemingly ``cohesive" group of voters~\cite{sfelkind}, who essentially form a solid coalition amongst themselves \cite{aziz2020expanding}.

We present two different classes of axioms for proportional representation in trichotomous preference domains the fundamental difference amongst whom is the definition of `cohesiveness' of a voter group. Essentially, Class I axioms provide a more liberal definition of cohesive representation while Class II axioms assert a stricter definition. As a result of this, the number of voter groups to be served representation in Class I axioms is higher than that in Class II axioms. We further propose new polynomial time executable voting rules and show their tendency to satisfy our proposed axioms through simulations. 

\subsection{Class I}

Our first stride in eliciting axioms for proportional representation in trichotomous voting domains is given next.  

\begin{remark}
We say that a group of voters $V'$ is worthy of justified representation if its size is at least the quantity that reflects uniform distribution of $k$ seats amongst $n$ voters i.e. $\frac{n}{k}$, while the voters in the group are at least as preferentially aligned as to support a set of candidates each of which is approved by at least one voter, but none is disapproved by any. In effect, the set of voters is said to be worthy and preferentially aligned if its size is at least $\frac{n}{k}$ and $|\cup_{i \in V'}A_i^+ \setminus \cup_{i \in V'}A_i^-| \neq \emptyset$.
\end{remark}

\begin{axiom}[Strong Preliminary Representation (SPR)]
A committee is said to satisfy \emph{Strong Preliminary Representation} if all subset of voters $V' \subseteq V$ of size $|V'| \geq \frac{n}{k}$ satisfy $|(\cup_{i \in V'}A_i^+) \setminus (\cup_{i \in V'}A_i^-)| \geq 1$ $\forall l \in [k]$,

the committee contains at least one candidate approved by at least one voter in $V'$ whilst not containing any unanimously disapproved candidate by the voters in the group. That is,
$$\exists i \in V': |W \cap A_i^+| \geq 1 \text{ and } |W \cap (\cap_{j \in V'}A_i^-)| = 0\ .$$
\end{axiom}

The definition above captures the idea that a large enough group with similarly aligned preferences should have at least one member who gets her favorable candidate in the committee but at the same time the unanimously disliked candidate should not feature in the committee. The committees satisfying this axiom take an all encompassing approach to the satisfaction of voter approval as well as voter resentment towards the committee. 

\begin{example}
Suppose there are $4$ voters with the following voter preferences over the set of candidates $a, b, c, d$ and $k = 2$. 
\begin{align*}
    1:& \hspace{8pt} \{d, a, b\} \succ c \succ e \\
    2:& \hspace{8pt} \{a, b\} \succ c \succ \{d, e\} \\
    3:& \hspace{10pt} a \succ \{b, c\} \succ \{d, e\} \\
    4:& \hspace{8pt} \{b, c\} \succ a \succ \{d, e\}
\end{align*}
\end{example}
The committee $\{a, b\}$ forms a committee adherent to the above axiom since there is at least one voter in each group who has a favorite candidate in the committee and there is no such group of voters which is large enough and find their commonly disliked candidate in the committee. It is useful to note that while \emph{Strong Preliminary Representation} seems to provide appropriate and in some ways a seemingly balanced allocation of candidates amongst the voters, it is indeed a strong notion and is unfortunately, not guaranteed to exist. 
\begin{example}
Suppose there are two voters and two candidates $\{v_1, v_2\}$ and $\{c_1, c_2\}$ respectively and the voter profile is as follows:
\begin{align*}
    1:& \hspace{8pt} c_1 \succ \{\} \succ c_2 \\
    2:& \hspace{8pt} c_1 \succ \{\} \succ c_2  
\end{align*}
Suppose that the committee size is $k = 2$, which means that there is only one committee possible which is $\{c_1, c_2\}$. This committee would not satisfy the axiom since there is at least one candidate in the committee that is despised by each voter, who in this case individually form a `deserving' and `cohesive' group of voters.  
\end{example}
In order to mitigate the non-existence guarantees of the above axiom, we propose its weakened version. In effect, we relax the mandatory debarring of the unanimously despised candidate from the committee. We assert that a large enough and cohesive group of voters is granted representation on the committee if there are at least some candidates in the committee that are approved by at least some voters in the voter group. 

\begin{axiom}[Weak Trichotomous Justified Representation (WTJR)]
A committee $W$ is said to follow \emph{Weak Trichotomous Justified Representation} if for all sets of voters $V'$ of size $|V'| \geq \frac{n}{k}$ having $|(\cup_{i \in V'}A_i^+) \setminus (\cup_{i \in V'}A_i^-)| \geq 1$ the following is satisfied:
$$|(\cup_{i \in V'}A_i^+)\cap W| \geq 1$$
\end{axiom}

\begin{axiom}[Weak Trichotomous Proportional Justified Representation (WTPJR)]
A committee $W$ is said to follow \emph{Weak Trichotomous Proportional Justified Representation} if for all sets of voters $V'$ of size $|V'| \geq l\frac{n}{k}$ for all $l \in [k]$ having $|(\cup_{i \in V'}A_i^+) \setminus (\cup_{i \in V'}A_i^-)| \geq l$ the following is satisfied:
$$|(\cup_{i \in V'}A_i^+)\cap W| \geq l$$
\end{axiom}
In effect, Weak Trichotomous Justified Representation is nothing but a special case for Weak Trichotomous Proportional Justified Representation with $l = 1$.  
As an instance, consider example 3.1 where the committee $\{a, b\}$ satisfies both the above axioms since for any large enough voter group, the set difference between the unions of approval and disapproval sets of voters contain $c_1$ as well as $c_2$.

\begin{remark}
A committee satisfying even \emph{Weak Trichotomous Justified Representation} is not guaranteed to exist for all voter profiles. For instance, if there are five voters and three candidates $\{a, b, c\}$ and $k = 2$ with the following voter profile:
\begin{align*}
    1:& \hspace{8pt} \{\} \succ \{a, b, c\}  \succ \{\} \\
    2:& \hspace{8pt} \{\}  \succ \{a, b, c\} \succ \{\} \\
    3:& \hspace{8pt} a \succ \{\} \succ \{b, c\} \\
    4:& \hspace{8pt} b \succ \{\} \succ  \{a, c\} \\
    5:& \hspace{8pt} c \succ \{\} \succ \{a, b\}
\end{align*}
In this case, whatever the $2$ size committee is, a voter group of voters $1, 2$ and any one of $3, 4, 5$ is unsatisfied. 
\end{remark}

On the contrary, we find that there always exists a committee that satisfies \emph{Weak Trichotomous Proportional Justified Representation} if there is some structure to the preferences of voters in the electorate. 
\begin{definition}{\emph{Decisive Electorate}}
An electorate is said to be decisive if none of the voters is indifferent about the available alternatives. Formally, $A_i^0 = \emptyset$ $\forall i \in V$.
\end{definition}

\begin{prop}
A committee $W$ satisfies Proportional Justified Representation in dichotomous preference domains if and only if it satisfies Weak Trichotomous Proportional Justified Representation in trichotomous preference domains in a decisive electorate. 
\end{prop}
\begin{proof}
Since the committee $W$ satisfies PJR, $|\cap_{i \in V'}A_i| \geq l$ $\forall l \in [k]$ and for all subsets of voters $V' \subseteq V$ such that $|V'| \geq l\frac{n}{k}$. Now since the electorate is decisive, $A_i^0 = \emptyset$ $\forall i \in V'$ for all subsets of voters $V' \subseteq V: |V'| \geq l\frac{n}{k}$. This means that for all such subsets of voters in the dichotomous domain,  $|\cap_{i \in V'}A_i| = |\cup_{i \in V'}A_i \setminus (C \setminus \cup_{i \in V'}A_i)| \geq l$, in effect, the intersection of approval sets of all voters in a deserving and cohesive group $V'$ is equal to the union of the approval sets minus the union of disapproval sets of all voters in the group. When projected in trichotomous domain, since this electorate is decisive, $A_i = A_i^+$ and $C \setminus A_i = A_i^-$ and $A_i^0 = \emptyset$. This means that all sets of voters that are assured representation in PJR are also guaranteed to have representation in WTPJR. Since the approval sets of voters under trichotomous and dichotomous domains are the same in a decisive electorate, the amount of representation that they get is also the same, i.e. $|\cup_{i \in V'}A_i \cap W| = |\cup_{i \in V'}A_i^+ \cap W| \geq l$. Since we have established the fact that the approval and disapproval sets of voters in dichotomous and trichotomous preference domains are equal, the sets of voters of a decisive electorate that get representation in the committee due to WTPJR are the same as those subsets of voters that get representation due to PJR and respectively get the same representation. Therefore, in a decisive electorate, a committee that satisfies WTPJR also satisfies PJR.   
\end{proof}

We define a still weaker axiom with the objective that it still provides at most as much representation to large enough voter groups as the previously given stronger axioms give.

\begin{axiom}[Weak Ambivalent Representation (WAR)]
A committee $W$ satisfies \emph{WAR} if for every group of voters $V'$ of size $V' \geq l\frac{n}{k}$ for some $l \in [k]$ the following condition is satisfied:
$$|(\cup_{i \in V'}A_i^+) \setminus (\cup_{i \in V'}A_i^-)| \geq l \implies |((\cup_{i \in V'}A_i^+) \cup (\cap_{i \in V'}A_i^0)) \cap W| \geq l$$ 
\end{axiom}

Intuitively, the definition says that the committee should have at least $l$ candidates from the set of candidates defined by the union of all the approved candidates and the unanimously `not cared about' candidates by the voters. In a way, the definition further degrades the utility that the voters in the voter groups are entitled to in the previous axioms. \\
Although the above definition is weaker than the previous axioms, there are still profiles for which there does not exist a committee such that it follows WAR. For instance, there is no committee that satisfies WAR if the preference profile is the same as mentioned in remark 3 i.e. the voter group formed by the first two voters and any one of the other three voters is always dis-satisfied with any committee formed out of the three candidates.

\begin{axiom}[Weakest Axiom (WA)]
A committee $W$ satisfies \emph{weakest axiom} if for every subset of voters $V' \subseteq V$ of size $|V| \geq l\frac{n}{k}$ for all $l \in [k]$, the following implication stands true:
$$|((\cup_{i \in V'}A_i^+) \setminus (\cup_{i \in V'}A_i^-))| \geq l \implies |((\cup_{i \in V'}A_i^+) \cup (\cup_{i \in V'}A_i^0))\cap W| \geq l$$
\end{axiom}

This is the most liberal axioms that we propose in Class I axioms. Intuitively, a committee satisfies this axiom if at least some every large enough group (same as defined before) get some candidates from their approval as well as indifference classes. Unfortunately, even committees satisfying this axiom are not guaranteed to exist, but as we present in the following sections, there is a high probability for a committee computed by some voting rules to satisfy this axiom.   \\

\subsection{Class II}

All the Class I axioms proposed have been in the spirit that the voters form cohesive groups with other voters even if there is a candidate about whom one of them is in approval while the other has neutral feelings as a result of which there are a high number of deserving and cohesive voter groups formed leading to an obvious difficulty in accommodating every voter's choice. In the definitions described in this section, we do away with this definition of cohesiveness to a less accommodating version but strengthen the amount of representation that the voters in the voter groups get in the committee.

\begin{axiom}[New Cohesiveness Representation (NCR)]

A committee satisfies \emph{New Cohesiveness Representation} if for no set of voters of size $|V'| \geq l\frac{n}{k}$ for all $l \in [1, k]$, if $|\cap_{i \in V'}A_i^+| \geq l$ then $|(\cap_{i \in V'}A_i^+) \cap W| < l$.
\end{axiom}

\begin{prop}
For a given committee size k, a committee satisfying \emph{New Cohesiveness Representation} always exists.
\end{prop}
\begin{proof}
 We present a constructive proof, inspired by the proof of proposition 3.7 presented in \cite{propPB}, which we specialize to the settings of multi-winner voting. Let's say, the algorithm iterates over $l'$ with its initial value $k$ and it terminates as soon as the value of $l'$ reaches $0$. Initially $A' = V$, where $A'$ acts as a container set for the voters who still remain to be served representation in the committee and let $W = \emptyset$. The algorithm greedily satisfies voter groups by providing them with representation in the committee and once they have been represented, they are removed from consideration to provide representation. At the beginning of every iteration, the condition $|W| + l' \leq k$ in order to ensure the committee size remains bounded by $k$; if this condition fails in any iteration, reduce $l'$ by $1$ and continue to the next iteration. Then, let $C^*$ denote the set of subsets of candidates of size $l'$. 
 $$C^* = \{C' \subseteq C: |C'| = l'\}$$
 If $C^*$ is empty, this means that there are no subsets of candidates of size exactly $l'$ and we cannot satisfy any subset of voters of size at least $l'\frac{n}{k}$ at this juncture of the algorithm, so we reduce $l'$ by $1$ and continue to the next iteration. If not so, for each subset of candidates $C' \in C^*$, define $A^+(C')$ as follows: 
 $$A^+(C') = \{i \in A': C' \subseteq A_i^+\}$$
 select the voter set $A^+(C')$ of maximal size, and check if $A^+(C') \geq l\frac{n}{k}$; if not so, decrease $l'$ by $1$ and continue to the next iteration, since the consequence of this step is that there is no subset of unrepresented voters which approves $l'$ candidates and has a size of at least $l'\frac{n}{k}$. But, if the condition is indeed true, set $W \rightarrow W  \cup \text{ }C'$ and $A' \rightarrow A' \setminus A^+(C')$. Preserving the value of $l'$ as is in pursuit of finding another subset of candidates of size $l'$ which is supported by correspondingly large group of unsatisfied voters and continue to the next iteration. If the algorithm reaches the point where $l' = 0$ but $|W| < k$, it means that there are no subsets of voters that are large and cohesive enough to have a candidate in the committee. In that case, we arbitrarily add candidates in the committee in order to fill it to its size, $k$. \\
 Now we prove the correctness of the algorithm. Suppose that a committee that is computed by the above algorithm $W$ does not satisfy \emph{New Cohesiveness Representation}, which means that there is at least one set of voters $V'$such that its size is $V' \geq l'\frac{n}{k}$ and $|\cap_{i \in V'}A_i^+| \geq l'$ but $|\cap_{i \in V'}A_i^+ \cap W| < l'$. Every candidate in the committee $W$ represents a group of at least $\frac{n}{k}$ voters, each of which once granted representation is not entertained further. This means that the number of voters getting representation in this committee is 
 $$|W|\cdot(\frac{n}{k}) = n$$
 thus, there does not exist any $l'$ such that a set of voters $V'$ not have $l' \leq |V'|k/n$ candidates to represent it in the committee. Thus, the contradiction.
\end{proof}
The above axiom promises representation to cohesive groups of voters from exactly that set of voters about which they mutually agree. 
We further present a weaker version of the \emph{Weaker New Cohesiveness Representation} whereby we weaken the amount of representation provided to the voters in a cohesive and deserving group. In some ways, this axiom mimics Proportional Justified Representation in dichotomous preferences. 

\begin{axiom}[Weaker New Cohesiveness Representation (WNCR)]
A committee satisfies \emph{New Cohesiveness Representation} if for no set of voters of size $|V'| \geq l\frac{n}{k}$ for all $l \in [1, k]$, if $|\cap_{i \in V'}A_i^+| \geq l$ then $|(\cup_{i \in V'}A_i^+) \cap W| < l$.
\end{axiom}

\begin{prop}
There exists a polynomial time algorithm that determines a committee that satisfies \emph{Weaker New Cohesiveness Representation}.
\end{prop}
The proof of this proposition can be argued to by asserting the equivalence between Weaker New Cohesiveness Representation in trichotomous preferences and Proportional Justified Representation in dichotomous preferences. This axiom mimics Proportional Justified Representation since the merging of disapproval and indifference sets of voters in the electorate would neither change the nature of the large enough voting groups who deserve representation nor would it change the nature of representation that they receive in the committee otherwise. Since we know that well known polynomial time rules like Sequential Phragmen's Rule satisfy Proportional Justified Representation \cite{phragmen}, we conclude that committees satisfying Weaker New Cohesiveness Representation always exist.

\begin{figure}
    \centering
    \includegraphics[scale = 0.9]{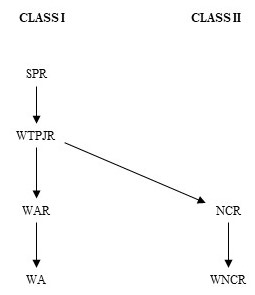}
    \caption{Relationship between the proposed Class I and Class II axioms}
    \label{fig:my_label}
\end{figure}
 
\section{Trichotomous Voting Rules}
In this section we present trichotomous versions of popular scoring rules especially with the objective of finding the most suitable rule providing proportional representation. While some voting rules have been taken from results cited correspondingly, we propose the approximate variants of trichotomous versions of Chamberlin Courant and PAV respectively and to the best of our knowledge, use Droop-STV in trichotomous settings. 
\subsection{Trichotomous Chamberlin Courant Rule} 
The $\alpha$-CC (Chamberlin Courant) rule is computed as follows. A satisfaction function for a committee $W \subseteq C$ such that $|W| \geq k$ is defined as follows: 

\[ sat_{\alpha-CC}(v, W) = 
    \begin{cases} 
      0 & |A_v^+ \cap W| - |A_v^- \cap W| < \alpha \\
      1 & otherwise
   \end{cases}
\]
The rule $\alpha$-CC finds a committee that maximizes $\Sigma_vsat_{\alpha-CC}(v, W)$, that is, the committee in which the highest number of voters have difference between the number of approved and disapproved candidates at least $\alpha$ in it. Since this rule is nothing but a slight variant of dichotomous $\alpha$-CC, computing a committee using this rule is NP-Hard \cite{CCAV} \cite{Zhou}.

\subsection{Trichotomous Proportional Approval Voting}
The TPAV score of a voter for a committee $W$ is calculated as follows; declare satisfaction and dis-satisfaction functions as follows:

\[ sat_{TPAV}(v, W) = 
    \begin{cases} 
      0 & \text{if } |A_v^+ \cap W| = 0 \\
      \Sigma_{p = 1}^{|A_v^+ \cap W|} \frac{1}{p} & otherwise
   \end{cases}
\]
\[ dissat_{TPAV}(v, W) = 
    \begin{cases} 
      0 & \text{if } |A_v^- \cap W| = 0 \\
      \Sigma_{p = 1}^{|A_v^- \cap W|} \frac{1}{p} & otherwise
   \end{cases}
\]
The committee that maximizes $\Sigma_{v}(sat_{TPAV}(v, W) - dissat_{TPAV}(v, W))$ is the one that is selected by the rule as the winner. From remark \ref{remark:remark1}, since TPAV is at least as hard as dichotomous PAV, winner determination in TPAV is NP-Hard. \cite{Zhou}

\subsection{Sequential Trichotomous CC and Trichotomous PAV}
The sequential variants of these rules can be generalized as follows; we start with an empty committee $W = \emptyset$ and iteratively add a candidate  $c$ in every iteration till $|W| < k$ such that $\Sigma_vsat_{\alpha-CC}(v, W \cup c)$ and $\Sigma_{v}(sat_{TPAV}(v, W \cup c) - dissat_{TPAV}(v, W \cup c))$ are maximized for Sequential Trichotomous $\alpha$-CC and Sequential Trichotomous PAV respectively. Note that at every iteration, the candidate added in the committee is deleted from the set of available candidates to be added in the committee. Generally, a candidate $c$ is added to the committee $W$ for a rule $\mathcal{R}$, where $\mathcal{R} \in \{\text{Sequential TCC, Sequential TPAV}\}$ if $\Sigma_v(sat_\mathcal{R}(v, W \cup c) - dissat_\mathcal{R}(v, W \cup c))$ is maximized, where $dissat_{TCC}(v, W) = 0$ $\forall W \in V^{k}$.

\subsection{Trichotomous Sequential Monroe}
There are several relevant approximation algorithms for the Monroe rule \cite{Lu10budgetedsocial} \cite{skowronMonroe}.
Here, we adapt Algorithm A proposed in \cite{skowronMonroe} to trichotomous settings. We proceed in $k$ steps, greedily building the committee $W$ by adding a candidate $c$ at every iteration. Define a satisfaction function $\Gamma(c): C \rightarrow V^{\Ceil{\frac{n}{k}}}$ such that for a candidate $c$ not yet added in the committee and $|\Gamma(c)| = \Ceil{\frac{n}{k}}$. Explicitly, $\Gamma(c)$ returns the set of voters such that $\forall i \in \Gamma(c) \text{ and } j \in V \setminus \Gamma(c), \text{ } pos_c(i) > pos_c(j)$ i.e. we choose the top $\Ceil{\frac{n}{k}}$ voters who have the highest positional score for the candidate $c$. Induct the candidate $c$ in the committee if the sum of positional scores of voters in $\Gamma(c)$ is the highest and remove the voters in $\Gamma(c)$ from the set of unsatisfied voters and $c$ from the set of available candidates. In cases where the number of unrepresented voters becomes $\emptyset$, we randomly add candidates into the committee $W$ until $|W| < k$. The approximation ratio is guaranteed to be $ 1 - \frac{k - 1}{2(m - 1)} - \frac{H_k}{k}$ where $m$ is the number of candidates and $H_k$ is the $k^{th}$ harmonic number \cite{skowronMonroe}. 

\subsection{Droop-Standard Transferable Vote}
We adapt this rule from the family of STV rules mentioned in \cite{aziz2020expanding}. We greedily construct a committee $W$ in the following manner; find the candidate $c$ with maximum plurality score (the important step here is to break ties randomly, not lexicographically) and compare it with the Droop quota i.e. $\Floor{\frac{n}{k + 1}} + 1$ and add the candidate to the committee if its score is greater than the quota, removing it from the set of candidates available for induction to the committee. Otherwise if the maximum score is less than the quota, we again break ties randomly, which is an important step for winner determination and remove that candidate from the list of available candidates. At the end of each iteration, we delete the candidate in focus, $c$ from the preference order of every voter and move on to the next iteration if $|W| < k$.

\section{Experimental Setup}
We conduct a series of experiments, each of which in itself has been conducted a number of times and hence we present the average results scaled so that optimum rule gives out a probability at most $1$. We conduct these experiments in order to quantitatively show which voting rule is the best suited for trichotomous settings. In order to do this, we consider $4$ of the well known polynomial time executable voting rules mentioned above and find the probability that the output of each of these voting rule satisfies our proposed axioms by generating  $10,000$ trichotomous voter preference profiles randomly.

We generate the voter preference profiles using the Impartial Culture model of preferences in which the weak preference order for a voter is drawn uniformly at random from the set of weak preference orders defined over the set of candidates $C$. Additionally, in each of the voter profiles, the number of voters and the number of candidates are picked randomly ranging from $4 \text{ to } 20$ and $1 \text{ to } 15$ respectively. \\
We use the Impartial Culture of Voting since it is a standard method of randomized profile generation to study voting mechanisms. As mentioned earlier, we divide the axioms in two different classes i.e. Class I and Class II. For every axiom presented, we find the probability that our chosen voting rules select a committee that satisfies it.   

\begin{table}

\begin{tabular}{|c|c|c|c|c|c|}
\hline
 Voting Rules & \multicolumn{3}{|c|}{Class I} & \multicolumn{2}{|c|}{Class II} \\
 \hline
 & WA & WAR & WTPJR & WNCR & NCR\\
 \hline
Sequential Monroe  & 0.9996  & 0.992 & 0.9878 & 0.996 & 0.769 \\
\hline
Sequential $\alpha$-Chamberlin Courant  & 0.9995 & 0.984 & 0.9696 & 0.9874 & 0.788 \\
\hline
Multi-winner STV  & 0.9998 & 0.993 & 0.99 & 0.9974 & 0.789\\
\hline
Sequential PAV & 0.9997 & 0.9934 & 0.9854 & 0.9928 & 0.8058 \\
\hline
\end{tabular}

\caption{Probabilities of voting rules satisfying Class I and Class II axioms over 10,000 randomly generated profiles with $|V| \in [4, 20]$ and $|C| \in [1, 15]$ between 1 and 15 and $k \in [1, |C| - 1]$ }
\label{Table 1}

\end{table}

\section{Discussion of Experimental Results}

We present our results in Table 1. We present a characterization of both Class I and Class II rules and further draw a contrast between the two and reason about the kind of results mentioned for both the classes correspondingly.

\subsection{Class I}
Our class $1$ axioms are based on the antecedent of getting representation being very strong i.e. there are larger number of large enough groups qualifying for representation as compared to those in the class $2$ axioms. In effect, all our sequential voting rules produce every good results with almost all entries in the table being close to 98\%. 

Amongst the Class I axioms, the strength of axioms increases as follows; WA, WAR and WTPJR. Due to a very weak nature of WA which allows voter groups getting representation from the union of approval and indifference classes, all voting rules almost always produce a committee that satisfies WA, though there are some exceptions in cases where the ratio $k/m$ is high. The data reflects a careful weakening of the stronger axioms WAR and WTPJR where the probability of finding a committee that satisfies these axioms is marginally less than the weakest version and hence these also emerge to be quite suitable axioms for evaluation of the quality of a committee determined by these polynomial time computable voting algorithms.

\subsection{Class II}

Our Class II axioms reveal an approach towards proportional representation that is commonly taken by Proportional Justified Representation \cite{PJR}. Hence, our voting rules tweaked for trichotomous ballots determine committees that satisfy axioms of this class with lesser probability as compared to the Class I axioms. NCR bears similarity to trichotomous version of Strong-BPJR \cite{propPB} presents a very strong notion of representation in the committee and hence does not always exist \cite{JR}. This is instantiated by the figures in table 1 where the best rule for this axiom is Sequential PAV while the worst being Sequential Monroe. Interestingly, the exact order is not repeated in the case when WNCR is taken as the axiom to be tested though what remains common is that Sequential $\alpha$-CC remains the voting rule which fares the worst. 

\subsection{Analysis}

The general trend in the fitness of voting rules for trichotomous settings turns out to be that Droop-STV produces the best results while Sequential $\alpha$-Chamberlin Courant rule produces the worst results. This re-asserts the usefulness of STV in electing a proportionally representative committee in dichotomous preferences as well \cite{aziz2020expanding} \cite{propMultiwinner} and also asserts the fact that Chamberlin Courant is not the best suited voting rule for proportional representation however impeccable it is for the selection of diverse committees \cite{Nimrod1}. Additionally, sequential Monroe and Sequential PAV also provide good results, although a little less than Droop-STV. In general, for the proposed Class I the following order of fittest voting rules for randomized preference ballot generation is found; Droop-STV, Sequential PAV, Sequential Monroe, Sequential $\alpha$-Chamberlin Courant. Alternatively, for Class II axioms, which take a little more cognizance of proportional representation axioms for dichotomous preferences \cite{PJR} \cite{JR}, it is correct to say that Sequential $\alpha$-CC is not the most suitable rule for randomly generated trichotomous voter profiles and while Droop-STV fares the best in weaker axioms of the class, it does not do so for the stronger axioms. Sequential-PAV is a good rule for those cases when a strong variant of axioms of this is preferred while other presented rules might not be the as effective as they are in with the weaker axioms.

\section{Outlook}

To be able to accommodate negative feelings, we have proposed two broad classes of axioms for proportional representation in trichotomous preference domains, which provide greater flexibility to the voters than approval ballots. We study and propose some voting rules for such input formats and show by simulations that these rules adhere to our axioms to a large extent. Therefore, we provide strong basis for the design of polynomial time computable voting rules taking input as trichotomous preferences, which select committees that satisfy certain proportionality axioms to this setting.
We mention some future research directions below.

\paragraph{Different Axioms and Rules}
Here we concentrated on adaptations of JR-style axioms to our setting. Naturally, there are other ways to approach proportionality, hopefully also giving rise to a richer landscape of voting rules for this setting.

\paragraph{Participatory Budgeting}
One application in which negative feelings are quite prominent is participatory budgeting, in which some projects might be perceived as hurting certain voter groups (e.g., building a sports stadium causes traffic jams, building a bus station causes pollution, etc.).
Currently such negative feelings are not taken into account while selecting project bundles for participatory budgeting; lifting the results of our paper to this important usecase is of theoretical as well as practical importance.

\bibliography{bibliography.bib}

\begin{thebibliography}{10}

\bibitem{Conitzer}
Vincent Conitzer.
\newblock Making decisions based on the preferences of multiple agents.
\newblock {\em Commun. ACM}, 53(3):84–94, March 2010.

\bibitem{Nimrod1}
Piotr Faliszewski, Arkadii~M. Slinko, and Nimrod Talmon.
\newblock Multiwinner voting: A new challenge for social choice theory.
\newblock In {\em Trends in Computational Social Choice}, 2017.

\bibitem{barbera2008choose}
Salvador Barber{\`a} and Danilo Coelho.
\newblock How to choose a non-controversial list with k names.
\newblock {\em Social Choice and Welfare}, 31(1):79--96, 2008.

\bibitem{propMultiwinner}
Edith Elkind, Piotr Faliszewski, Piotr Skowron, and Arkadii Slinko.
\newblock Properties of multiwinner voting rules.
\newblock In {\em Proceedings of the 2014 International Conference on
  Autonomous Agents and Multi-Agent Systems}, AAMAS ’14, page 53–60,
  Richland, SC, 2014. International Foundation for Autonomous Agents and
  Multiagent Systems.

\bibitem{articleProp}
Edith Elkind, Piotr Faliszewski, Piotr Skowron, and Arkadii Slinko.
\newblock Properties of multiwinner voting rules.
\newblock {\em 13th International Conference on Autonomous Agents and
  Multiagent Systems, AAMAS 2014}, 1, 06 2015.

\bibitem{SkowronLang}
Piotr Skowron, Piotr Faliszewski, and Jerome Lang.
\newblock Finding a collective set of items: From proportional
  multirepresentation to group recommendation.
\newblock In {\em Proceedings of the Twenty-Ninth AAAI Conference on Artificial
  Intelligence}, AAAI’15, page 2131–2137. AAAI Press, 2015.

\bibitem{LB15}
Tyler Lu and Craig Boutilier.
\newblock Value-directed compression of large-scale assignment problems.
\newblock In {\em Proceedings of the Twenty-Ninth AAAI Conference on Artificial
  Intelligence}, AAAI’15, page 1182–1190. AAAI Press, 2015.

\bibitem{Lu10budgetedsocial}
Tyler Lu and Craig Boutilier.
\newblock Budgeted social choice: From consensus to personalized decision
  making.
\newblock In Toby Walsh, editor, {\em {IJCAI} 2011, Proceedings of the 22nd
  International Joint Conference on Artificial Intelligence, Barcelona,
  Catalonia, Spain, July 16-22, 2011}, pages 280--286. {IJCAI/AAAI}, 2011.

\bibitem{sfelkind}
Luis S{\'a}nchez-Fern{\'a}ndez, Edith Elkind, Martin Lackner, Norberto
  Fern{\'a}ndez, Jes{\'u}s~A Fisteus, Pablo~Basanta Val, and Piotr Skowron.
\newblock Proportional justified representation.
\newblock In {\em Thirty-First AAAI Conference on Artificial Intelligence},
  2017.

\bibitem{JR}
Haris Aziz and Toby Walsh.
\newblock Justified representation in approval-based committee voting.
\newblock {\em CoRR}, abs/1407.8269, 2014.

\bibitem{PJR}
Luis~S{\'{a}}nchez Fern{\'{a}}ndez, Edith Elkind, Martin Lackner,
  Norberto~Fern{\'{a}}ndez Garc{\'{\i}}a, Jes{\'{u}}s~A. Fisteus, Pablo
  Basanta{-}Val, and Piotr Skowron.
\newblock Proportional justified representation.
\newblock {\em CoRR}, abs/1611.09928, 2016.

\bibitem{CCAV}
John~R. Chamberlin and Paul~N. Courant.
\newblock Representative deliberations and representative decisions:
  Proportional representation and the borda rule.
\newblock {\em The American Political Science Review}, 77(3):718--733, 1983.

\bibitem{betzler2013computation}
Nadja Betzler, Arkadii Slinko, and Johannes Uhlmann.
\newblock On the computation of fully proportional representation.
\newblock {\em Journal of Artificial Intelligence Research}, 47:475--519, 2013.

\bibitem{Aziz2014ComputationalAO}
Haris Aziz, Serge Gaspers, Joachim Gudmundsson, Simon Mackenzie, Nicholas
  Mattei, and Toby Walsh.
\newblock Computational aspects of multi-winner approval voting.
\newblock In {\em MPREF@AAAI}, 2014.

\bibitem{Alcantud}
Laruelle~A. Alcantud~J.C.R.
\newblock Dis\&approval voting : A characterization.
\newblock {\em Social Choice Welfare}, 43:1 -- 10, 2014.

\bibitem{Hillinger}
Claude Hillinger.
\newblock The case for utilitarian voting.
\newblock {\em Homo Oeconomicus}, 22(3), 2005.

\bibitem{brams}
Steven~J. Brams.
\newblock When is it advantageous to cast a negative vote?
\newblock In Rudolf Henn and Otto Moeschlin, editors, {\em Mathematical
  Economics and Game Theory}, pages 564--572, Berlin, Heidelberg, 1977.
  Springer Berlin Heidelberg.

\bibitem{yilmaz}
Mustafa~R Yilmaz.
\newblock Can we improve upon approval voting?
\newblock {\em European Journal of Political Economy}, 15(1):89 -- 100, 1999.

\bibitem{Falsenthal}
Dan~S. Felsenthal.
\newblock On combining approval with disapproval voting.
\newblock {\em Behavioral Science}, 34(1):53--60, 1990.

\bibitem{smaoui}
Hatem Smaoui and Dominique Lepelley.
\newblock The voting system {\ `e} me by note {\` a} three levels: {\ 'e} study
  of a new voting system.
\newblock {\em Journal of {\ 'e} political economy}, 123(6):827--850, 2013.

\bibitem{bestworst}
José~Luis García-Lapresta, A.~A.~J. Marley, and Miguel Martínez-Panero.
\newblock Characterizing best-worst voting systems in the scoring context.
\newblock {\em Social Choice and Welfare}, 34(3):487--496, 2010.

\bibitem{baum}
Dorothea Baumeister and Sophie Dennisen.
\newblock Voter dissatisfaction in committee elections.
\newblock In {\em Proceedings of the 2015 International Conference on
  Autonomous Agents and Multiagent Systems}, AAMAS ’15, page 1707–1708,
  Richland, SC, 2015. International Foundation for Autonomous Agents and
  Multiagent Systems.

\bibitem{Zhou}
Aizhong Zhou, Yongjie Yang, and Jiong Guo.
\newblock Parameterized complexity of committee elections with dichotomous and
  trichotomous votes.
\newblock In {\em Proceedings of the 18th International Conference on
  Autonomous Agents and MultiAgent Systems}, AAMAS ’19, page 503–510,
  Richland, SC, 2019. International Foundation for Autonomous Agents and
  Multiagent Systems.

\bibitem{Ouafdi}
\&~Smaoui~H. El~Ouafdi~A., Lepelley~D.
\newblock On the condorcet efficiency of evaluative voting (and other voting
  rules) with trichotomous preferences.
\newblock {\em Ann Oper Res}, 2020.

\bibitem{aziz2020expanding}
Haris Aziz and Barton~E Lee.
\newblock The expanding approvals rule: Improving proportional representation
  and monotonicity.
\newblock {\em Social Choice and Welfare}, 54(1):1--45, 2020.

\bibitem{propPB}
Haris Aziz, Barton~E. Lee, and Nimrod Talmon.
\newblock Proportionally representative participatory budgeting: Axioms and
  algorithms.
\newblock In {\em Proceedings of the 17th International Conference on
  Autonomous Agents and MultiAgent Systems}, AAMAS ’18, page 23–31,
  Richland, SC, 2018. International Foundation for Autonomous Agents and
  Multiagent Systems.

\bibitem{phragmen}
Markus Brill, Rupert Freeman, Svante Janson, and Martin Lackner.
\newblock Phragm\'{e}n’s voting methods and justified representation.
\newblock In {\em Proceedings of the Thirty-First AAAI Conference on Artificial
  Intelligence}, AAAI’17, page 406–413. AAAI Press, 2017.

\bibitem{skowronMonroe}
Piotr Skowron, Piotr Faliszewski, and Arkadii Slinko.
\newblock Achieving fully proportional representation: Approximability results.
\newblock {\em Artificial Intelligence}, 222:67--103, 2015.

\end{thebibliography}
\bibliographystyle{unsrt}
\end{document}